\pdfoutput=1
\documentclass[10pt,twocolumn,conference]{IEEEtran}
\usepackage{amsmath,amsfonts,amssymb}
\usepackage{graphicx}
\usepackage{cite}
\usepackage{multirow}
\usepackage{algorithm}
\usepackage{algorithmic}

\newtheorem{theorem}{Theorem}
\newtheorem{definition}{Definition}
\newtheorem{lemma}{Lemma}
\newtheorem{corollary}{Corollary}

\newtheorem{example}{Example}

\newcommand{\set}[1]{\mathcal{#1}}
\newcommand{\defined}{\triangleq}
\newcommand{\Real}{{\mathbb{R}}}

\newcommand{\graph}{\set{G}}
\newcommand{\nodes}{\set{V}}
\newcommand{\edges}{\set{E}}
\newcommand{\N}{\set{N}}

\DeclareMathOperator{\indep}{\perp}
\newcommand{\X}{\set{X}}
\newcommand{\C}{\set{C}}
\newcommand{\R}{\set{R}}
\newcommand{\A}{\set{A}}\newcommand{\B}{\set{B}}
\newcommand{\GammaIn}{\Gamma^\text{In}}

\newcommand{\eht}[2]{#1 \to #2}
\newcommand{\tail}[1]{\mathrm{tail}(#1)}
\newcommand{\head}[1]{\mathrm{head}(#1)}
\newcommand{\sessions}{\set{S}}
\newcommand{\parent}{\pi}

\DeclareMathOperator{\An}{An}

\title{Network Coding Capacity: A Functional Dependence Bound}
\author{\authorblockN{Satyajit Thakor, Alex Grant and Terence Chan}
\authorblockA{Institute for Telecommunications Research\\
University of South Australia}
{satyajitsinh.thakor@postgrads.unisa.edu.au}}

\begin{document}
\maketitle

\begin{abstract}
  Explicit characterization and computation of the multi-source
  network coding capacity region (or even bounds) is long standing
  open problem. In fact, finding the capacity region requires
  determination of the set of all entropic vectors $\Gamma^{*}$, which
  is known to be an extremely hard problem. On the other hand,
  calculating the explicitly known linear programming bound is very
  hard in practice due to an exponential growth in complexity as a
  function of network size. We give a new, easily computable outer
  bound, based on characterization of all functional dependencies in
  networks.  We also show that the proposed bound is tighter than some
  known bounds.
\end{abstract}

\section{Introduction}

The network coding approach introduced in~\cite{AhlCai00,LiYeu03}
generalizes routing by allowing intermediate nodes to forward packets
that are coded combinations of all received data packets. This yields
many benefits that are by now well documented in the literature
\cite{YeuLi06,FraSol07,FraSol08,HoLun08}.  One fundamental open
problem is to characterize the capacity region and the classes of
codes that achieve capacity. The single session multicast problem
is well understood. In this case, the capacity region is characterized by
max-flow/min-cut bounds and linear network codes maximize throughput
\cite{LiYeu03}.

Significant complications arise in more general scenarios, involving
more than one session.  Linear network codes are not sufficient for
the multi-source problem~\cite{DouFre05,ChaGra08}.  Furthermore, a
computable characterization of the capacity region is still
unknown. One approach is to bound the capacity region by the
intersection of a set of hyperplanes  (specified by the network
topology and sink demands) and the set of entropy functions $\Gamma^*$
(inner bound), or its closure $\bar{\Gamma}^*$ (outer
bound)~\cite{Yeu02,SonYeu03,YeuLi06}.  An exact expression for the
capacity region does exist, again in terms of
$\Gamma^*$~\cite{YanYeu07}.  Unfortunately, this expression, or even
the bounds~\cite{Yeu02,SonYeu03,YeuLi06} cannot be computed in
practice, due to the lack of an explicit characterization of the set
of entropy functions for more than three random variables. In fact, it
is now known that $\Gamma^*$ cannot be described as the intersection
of finitely many half-spaces \cite{Mat07}. The difficulties arising
from the structure of $\Gamma^*$ are not simply an artifact of the way
the capacity region and bounds are written. In fact it has been shown
that the problem of determining the capacity region for multi-source
network coding is completely equivalent to characterization of
$\Gamma^*$ \cite{ChaGra08}.

One way to resolve this difficulty is via relaxation of the bound,
replacing the set of entropy functions with the set of polymatroids
$\Gamma$ (which has a finite characterization). In practice however,
the number of variables and constraints increase exponentially with
the number of links in the network, and this prevents practical
computation for any meaningful case of interest.

In this paper, we provide an easily computable relaxation of the LP
bound. The main idea is to find sets of edges which are determined by
the source constraints and sink demands such that the total capacity of these sets
bounds the total throughput. The resulting bound is tighter than the
network sharing bound \cite{YanYan06} and the bounds based on
information dominance \cite{HarKle06}.

Section~\ref{sec:background} provides some background on
pseudo-variables and pseudo entropy functions (which generalize
entropy functions) \cite{ChaGra08}. These pseudo variables are used to
describe a family of linear programming bounds on the capacity region
for network coding.  In Section~\ref{sec:FDG} we give an abstract
definition of a functional dependence graph, which expresses a set of
local dependencies between pseudo variables (in fact a set of
constraints on the pseudo entropy). Our definition extends that
introduced by Kramer \cite{Kra98} to accommodate
cycles. This section also provides the main technical
ingredients for our new bound. In particular, we describe a test for
functional dependence, and give a basic result relating local and
global dependence.  The main result is presented in
Section~\ref{sec:bound}.

\emph{Notation}: Sets will be denoted with calligraphic typeface,
e.g. $\X$. Set complement is denoted by the superscript $\X^c$ (where
the universal set will be clear from context). Set subscripts identify the
set of objects indexed by the subscript: $X_\A=\{X_a, a\in\A\}$. The
power set $2^\X=\{\A, \A\subseteq\X\}$ is the collection of all
subsets of $\X$. Where no confusion will arise, set union will be
denoted by juxtaposition, $\A\cup\B=\A\B$, and singletons will be
written without braces.

\section{Background}\label{sec:background}
\subsection{Pseudo Variables}
We give a brief revision of the concept of pseudo-variables,
introduced in \cite{ChaGra08}.  Let $\N=\{1,2,\dots,N\}$ be a finite
set, and let $X_\N=\{X_1,X_2,\dots,X_N\}$ be a ground set associated
with a   real-valued function $g:2^{X_\N}\mapsto\Real$ defined on
subsets of $X_\N$, with $g(\emptyset)=0$. We refer to the elements of
$X_\N$ as \emph{pseudo-variables} and the function $g$ as a
\emph{pseudo-entropy} function. Pseudo-variables and pseudo-entropy
generalize the familiar concepts of random variables and
entropy. Pseudo-variables do not necessarily take values, and there
may be no associated joint probability distribution. A pseudo-entropy
function may assign values to subsets of $X_\N$ in a way that is not
consistent with any distribution on a set of $N$ random variables.  A
pseudo-entropy function $g$ can be viewed as a point in a $2^N$
dimensional Euclidean space, where each coordinate of the space is
indexed by a subset of $X_\N$.

A function $g$ is called \emph{polymatroidal} if it satisfies the
polymatroid axioms.
\begin{align}
    g(\emptyset) &= 0 \label{poly:zero} \\
    g(X_\A) &\ge g(X_\B),\quad\text{if}\ \B\subseteq \A
    &&\text{non-decreasing} \label{poly:nondec}\\
    g(X_\A) + g(X_\B) &\ge g({X_{\A\cup\B}}) +
  g({X_{\A\cap\B}}) && \text{submodular} \label{poly:submod}
  \end{align}
  It is called \emph{Ingletonian} if it satisfies Ingleton's
  inequalities (note that Ingletonian $g$ are also polymatroids)
  \cite{GuiCha08}. A function $g$ is \emph{entropic} if it corresponds
  to a valid assignment of joint entropies on $N$ random variables,
  i.e. there exists a joint distribution on $N$ discrete finite random
  variables $Y_1,\dots,Y_N$ with $g(X_\A)=H(Y_\A),
  \A\subseteq\N$. Finally, $g$ is \emph{almost entropic} if there
  exists a sequence of entropic functions $h^{(k)}$ such that
  $\lim_{k\to\infty} h^{(k)} = g$.  Let
  $\Gamma^*\subset\bar{\Gamma}^*\subset\GammaIn\subset\Gamma$
  respectively denote the sets of all entropic, almost entropic,
  Ingletonian and polymatroidal, functions.

  Both $\Gamma$ and $\GammaIn$ are polyhedra. They can be expressed as
  the intersection of a finite number of half-spaces in $\Real^N$. In
  particular, every $g\in\Gamma$ satisfies,
  \eqref{poly:zero}-\eqref{poly:submod}, which can be expressed
  minimally in terms of
  $$N + \binom{N}{2}2^{N-2}$$ linear inequalities involving $2^N-1$
  variables \cite{Yeu02}. Each $g\in\GammaIn$ satisfies an additional
  $$\frac{1}{4}6^N-5^N+\frac{3}{2}4^N - 3^N+\frac{1}{4}2^N$$ linear
  inequalities \cite{GuiCha08}.


\begin{definition}\label{def:function}
  Let $\A,\B\subseteq\X$ be subsets of a set of pseudo-variables $\X$
  with pseudo-entropy $g$. Define\footnote{Note that this yields the
    chain rule for pseudo-entropies to be true by definition.}
  \begin{equation}
    g\left(\B\mid\A\right) \defined g\left(\A
      \B\right) - g(\A).\label{eq:chainrule}
 \end{equation}
  A pseudo-variable $X\in\X$ is said to be
  a \emph{function} of a set of pseudo-variables $\A\subseteq\X$ if
  $g\left(X \mid \A\right) = 0$.
\end{definition}
\begin{definition}\label{def:independent}
  Two subsets of pseudo-variables $\A$ and $\B$ are called
  \emph{independent} if $g\left(\A \B\right) = g(\A ) + g(\B)$,
  denoted by $\A \indep \B$.
\end{definition}

\subsection{Network Coding and Capacity Bounds}


Let the directed acyclic graph $\graph = (\nodes, \edges)$ serve as a
simplified model of a communication network with error-free
point-to-point communication links. Edges $e\in\edges$ have capacity
$C_e>0$.  For edges $e,f\in\edges$, write $\eht{f}{e}$ as shorthand
for $\head{f}=\tail{e}$.  Similarly, for an edge $f\in\edges$ and a
node $u\in\nodes$, the notations $\eht{f}{u}$ and $u\rightarrow f$
respectively denote $\head{f}=u$ and $\tail{f}=u$.

Let $\sessions$ be an index set for a number of multicast
sessions, and let $\{Y_{s}: s \in \sessions\}$ be the set of
source variables. These sources are available at the nodes
identified by the mapping $a : \sessions \mapsto \nodes$. Each
source may be demanded by multiple sink nodes, identified by the
mapping $b : \sessions \mapsto 2^{\nodes}$. Each edge $e\in\edges$
carries a variable $U_e$ which is a function of incident edge
variables and source variables.


\begin{definition}
  Given a network $\graph=(\nodes,\edges)$, with sessions $\sessions$, source
  locations $a$ and sink demands $b$, and a subset of pseudo-entropy
  functions $\Delta\subset\Real^{2^{|\sessions|+|\edges|}}$ on
  pseudo-variables $Y_\sessions,U_\edges$, let $\R(\Delta) =
  \{(R_s,s\in\sessions)\}$ be the set of source rate tuples for which
  there exists a $g\in\Delta$ satisfying
  \begin{align}
    \indep_{s\in\sessions} & \; Y_s \tag{$\C_1$} \\
    g\left(U_e \mid \{Y_s : {\eht{a(s)}{e}}\}, \{U_f: \eht{f}{e}\}
    \right) & = 0,   e\in\edges \tag{$\C_2$}\\
    g\left(Y_s \mid U_e : {\eht{e}{u}}\right) & = 0,
    u\in b(s) \tag{$\C_3$}\\
    g(U_e) & \le C_e,  e\in\edges \tag{$\C_4$} \\
    g(Y_s) & \ge R_s, s\in\sessions \notag
  \end{align}
\end{definition}
It is known that $\R({\Gamma^*)}$ and $\R({\bar{\Gamma}^*)}$ are inner and
outer bounds for the set of achievable rates (i.e. rates for which
there exist network codes with arbitrarily small probability of
decoding error).

It is known that $\R(\Gamma)$ is an outer bound for the set of
achievable rates \cite{Yeu02}. Similarly, $\R(\GammaIn)$ is an outer
bound for the set of rates achievable with linear network codes
\cite{ChaGra08}. Clearly $\R(\GammaIn) \subset \R(\Gamma)$.
The sum-rate bounds induced by $\R(\GammaIn)$ and  $\R(\Gamma)$
can in principle be computed using linear programming, since
they may be reformulated as
\begin{equation}\label{eq:LP}
  \begin{split}
    \max \sum_{s\in\sessions} g(Y_s) \quad
    \text{subject to} \\
    g \in \C_1 \cap \C_2 \cap \C_3 \cap \C_4 \cap \Delta
  \end{split}
\end{equation}
where $\Delta$ is either $\Gamma$ or $\GammaIn$, and $\C_1, \C_2,
\C_3, \C_4$ are the subsets of pseudo-entropy functions satisfying the
so-labeled constraints above. Clearly the constraint set $\C_1 \cap
\C_2 \cap \C_3 \cap \C_4 \cap \Delta$ is linear.

One practical difficulty with computation of \eqref{eq:LP} is the
number of variables and the number of constraints due to $\Gamma$ (or
$\GammaIn$), both of which increase exponentially with $|\edges|$. The
aim of this paper is to find a simpler outer bound. One approach is to
use the functional dependence structure induced by the network
topology to eliminate variables or constraints from $\Gamma$
\cite{GriGra08}.  Here we will take a related approach, that directly
delivers an easily computable bound.

\section{Functional Dependence Graphs}\label{sec:FDG}
\begin{definition}[Functional Dependence Graph]\label{def:FDG}
  Let $\X=\{X_1,\dots,X_N\}$ be a set of pseudo-variables with
  pseudo-entropy function $g$. A directed graph
  $\graph=(\nodes,\edges)$ with $|\nodes|=N$ is called a
  \emph{functional dependence graph} for $\X$ if and only if for all
  $i=1,2,\dots,N$
  \begin{equation}\label{eq:localdependence}
    g\left(X_i \mid \{ X_j : (j,i)\in\edges \}\right) = 0.
  \end{equation}
\end{definition}
With an identification of $X_i$ and node $i\in\nodes$, this Definition
requires that each pseudo-variable $X_i$ is a function (in the sense
of Definition \ref{def:function}) of the pseudo-variables associated
with its parent nodes. To this end, define
\begin{equation*}
  \parent(i) = \{j\in\nodes: (j,i)\in\edges\}.
\end{equation*}
Where it does not cause confusion, we will abuse notation and identify
pseudo-variables and nodes in the FDG, e.g. \eqref{eq:localdependence}
will be written
  $g\left(i \mid \parent(i)\right) = 0$.

  Definition \ref{def:FDG} is more general than the functional
  dependence graph of \cite[Chapter 2]{Kra98}. Firstly, in our
  definition there is no distinction between source and non-source
  random variables. The graph simply characterizes functional
  dependence between variables. In fact, our definition admits cyclic
  directed graphs, and there may be no nodes with in-degree zero
  (which are source nodes in \cite{Kra98}). We also do not require
  independence between sources (when they exist), which is implied by
  the acyclic constraint in \cite{Kra98}. Our definition of an FDG
  admits pseduo-entropy functions $g$ with \emph{additional}
  functional dependence relationships that are not represented by the
  graph. It only specifies a certain set of conditional
  pseudo-entropies which must be zero.  Finally, our definition holds
  for a wide class of objects, namely pseudo-variables, rather than
  just random variables.

Clearly a functional dependence graph in the sense of
\cite{Kra98} satisfies the conditions of Definition \ref{def:FDG}, but
the converse is not true. Henceforth when we refer to a functional
dependence graph (FDG), we mean in the sense of Definition
\ref{def:FDG}. Furthermore, an FDG is \emph{acyclic} if $\graph$ has
no directed cycles. A graph will be called \emph{cyclic} if every node
is a member of a directed cycle.\footnote{In this paper we do not
  consider graphs that are neither cyclic or acyclic.}


Definition \ref{def:FDG} specifies an FDG in terms of local dependence
structure. Given such local dependence constraints, it is of great
interest to determine all implied functional dependence relations.
In other words, we wish to find all sets $\A$ and $\B$ such that
$g(\A \B)=g(\A)$.

\begin{definition}\label{def:fd}
  For disjoint sets $\A,\B\subset\nodes$ we say $\A$ determines $\B$
  in the directed graph $\graph=(\nodes,\edges)$, denoted
  $\A\rightarrow\B$, if there are no elements of $\B$ remaining after
  the following procedure:

  Remove all edges outgoing from nodes in $\A$ and
  subsequently remove all nodes and edges with no incoming edges and
  nodes respectively.
\end{definition}
For a given set $\A$, let $\phi(\A)\subseteq\nodes$ be the set of nodes
deleted by the procedure of Definition~\ref{def:fd}. Clearly
$\phi(\A)$ is the largest set of nodes with $\A\rightarrow\phi(\A)$.

\begin{lemma}[Grandparent lemma]\label{lem:grandparent}
  Let $\graph=(\nodes,\edges)$ be a functional dependence graph for a
  polymatroidal pseudo-entropy function $g\in\Gamma$. For any
  $j\in\nodes$ with $i\in\parent(j)\neq\emptyset$
\begin{equation*}
  g\left(j\mid \parent(i), \parent(j) \setminus i \right) = 0.
\end{equation*}
\end{lemma}
\begin{proof}
  By hypothesis, $g(k\mid \parent(k))=0$ for any
  $k\in\nodes$. Furthermore, note that for any $g\in\Gamma$,
  conditioning cannot increase pseudo-entropy\footnote{This is a
    direct consequence of submodularity, \eqref{poly:submod}.} and
  hence $g(k\mid\parent(k),\A)=0$ for any $\A\subseteq\nodes$. Now
  using this property, and the chain rule
  \begin{align*}
    0 &= g(j\mid\parent(j)) \\
    &= g(j \mid \parent(j), \parent(i)) \\
    &= g(j, \parent(j), \parent(i)) - g(\parent(j), \parent(i)) \\
    &= g(j, \parent(j)\setminus i, \parent(i))
- g(\parent(j), \parent(i)) \\
  &= g(j, \parent(j)\setminus i, \parent(i)) -
  g(\parent(j)\setminus i, \parent(i))
\\
  &= g(j\mid \parent(i), \parent(j) \setminus i ).
  \end{align*}
\end{proof}
We emphasize that in the proof of Lemma~\ref{lem:grandparent} we have
only used the submodular property of polymatroids, together with the
hypothesized local dependence structure specified by the FDG.

Clearly the lemma is recursive in nature. For example, it is valid for
$g(j\mid \parent(j) \setminus i, \parent(i)\setminus k, \parent(k))=
0$ and so on. The implication of the lemma is that a pseudo-variable
$X_j$ in an FDG is a function of $X_\A$ for any $\A\subset\nodes$ with
$\A\rightarrow j$.

\begin{theorem}
Let $\graph$ be a functional dependence graph on the pseudo-variables
$\X$ with polymatroidal pseudo-entropy function $g$. Then for disjoint
subsets $\A,\B\subset\nodes$,
\begin{equation*}
\A \rightarrow \B \implies g(\B \mid \A) = 0.
\end{equation*}
\end{theorem}
\begin{proof}
  Let $\A\rightarrow\B$ in the FDG $\graph$. Then, by Definition
  \ref{def:fd} there must exist directed paths from some nodes in $\A$
  to a every node in $\B$, and there must not exist any directed path
  intersecting $\B$ that does not also intersect $\A$. Recursively
  invoking Lemma~\ref{lem:grandparent}, the theorem is proved.
\end{proof}

Definition~\ref{def:fd} describes an efficient graphical procedure to
find implied functional dependencies for pseudo-variables with local
dependence specified by a functional dependence graph $\graph$. It
captures the essence of the chain rule \eqref{eq:chainrule} for
pseudo-entropies and the fact that pseudo-entropy is non-increasing
with respect to conditioning (\ref{poly:submod}), which are the main
arguments necessary for manual proof of functional dependencies.

One application of Definition~\ref{def:fd} is to find a reduction of a
given set $\C$, i.e. to find a disjoint partition of $\C$ into $\A$
and $\B$ with $\A\rightarrow\B$, which implies
$g(\C)=g(\A\B)=g(\A)$. On the other hand, it also tells which sets are
\textit{irreducible}.
\begin{definition}[Irreducible set]
  A set of nodes $\B$ in a functional dependence graph is
  \emph{irreducible} if there is no $\A\subset\B$ with
  $\A\rightarrow\B$.
\end{definition}
Clearly, every singleton is irreducible. In addition, in an acyclic
FDG, irreducible sets are basic entropy sets in the sense of
\cite{GriGra08}.
In fact, irreducible sets generalize the idea of basic entropy sets to
the more general (and possibly cyclic) functional dependence graphs on
pseudo-variables.

\subsection{Acyclic Graphs}
In an acyclic graph, let $\An(\A)$ denote the set of ancestral nodes,
i.e. for every node $a\in\An(\A)$, there is a directed path from $a$
to some $b\in\A$.

Of particular interest are the maximal irreducible sets:
\begin{definition}
  An irreducible set $\A$ is \emph{maximal} in an acyclic FDG
  $\graph=(\nodes,\edges)$ if
  $\nodes\setminus\phi(\A)\setminus\An(\A) \defined (\nodes\setminus\phi(\A))\setminus\An(\A) =\emptyset$, and no proper
  subset of $\A$ has the same property.
\end{definition}
Note that for acyclic graphs, every subset of a maximal irreducible
set is irreducible. Conversely, every irreducible set is a subset of
some maximal irreducible set \cite{GriGra08}. Irreducible sets can be
augmented in the following way.
\begin{lemma}[Augmentation]\label{lem:augment}
  Let $A\subset\nodes$ in an acyclic FDG $\graph=(\nodes,\edges)$. Let
  $\B=\nodes\setminus\phi(\A)\setminus\An(\A)$.  Then $\A\cup\{b\}$ is
  irreducible for every $b\in \B$.
\end{lemma}
This suggests a process of recursive augmentation to find all maximal
irreducible sets in an acyclic FDG (a similar process of augmentation
was used in \cite{GriGra08}).  Let $\graph$ be a topologically
sorted\footnote{I.e. Order nodes such that if there is a directed edge
  from node $i$ to $j$ then $i\prec j$ \cite[Proposition
  11.5]{Yeu02}.} acyclic functional dependence graph
$\graph=(\{0,1,2,\dots\},\edges)$. Its maximal irreducible sets can be
found recursively via $\textbf{AllMaxSetsA}(\graph,\{\})$ in Algorithm
\ref{alg:augment}. In fact, $\textbf{AllMaxSetsA}(\graph,\A)$
finds all maximal irreducible sets containing $\A$.

\begin{algorithm}
\caption{\textbf{AllMaxSetsA}($\graph,\A$)}\label{alg:augment}
  \begin{algorithmic}
    \REQUIRE $\graph=(\nodes,\edges), \A\subset\nodes$
    \STATE
    $\B\leftarrow\nodes\setminus\phi(\A)\setminus\An(\A)$
    \IF{$\B\neq\emptyset$}
    \STATE Output $\{\textbf{AllMaxSetsA}(\graph,\A\cup \{b\} ): b\in\B \}$
    \ELSE
    \STATE Output $\A$
    \ENDIF
  \end{algorithmic}
\end{algorithm}

\subsection{Cyclic Graphs}
In cyclic graphs, the notion of a maximal irreducible set is modified
as follows:
\begin{definition}
  An irreducible set $\A$ is \emph{maximal} in a cyclic FDG
  $\graph=(\nodes,\edges)$ if $\nodes\setminus\phi(\A)=\emptyset$, and
  no proper subset of $\A$ has the same property.
\end{definition}
For cyclic graphs, every subset of a maximal irreducible set is
irreducible. In contrast to acyclic graphs, the converse is not
true. In fact there can be irreducible sets that are not maximal, and
are not subsets of any maximal irreducible set.
It is easy to show that
\begin{lemma}\label{lem:equal}
  All maximal irreducible sets have the same pseudo-entropy.
\end{lemma}
This fact will be used in development of our capacity bound for
network coding in Section \ref{sec:bound} below.  We are interested in
finding every maximal irreducible set for cyclic graphs. This may be
accomplished recursively via $\textbf{AllMaxSetsC}(\graph,\{\})$ in
Algorithm \ref{alg:AMS}. Note that in contrast to Algorithm
\ref{alg:augment}, $\textbf{AllMaxSetsC}(\graph,\A)$ finds all maximal
irreducible sets that \emph{do not} contain any node in $\A$.
\begin{algorithm}
  \caption{\textbf{AllMaxSetsC}($\graph,\A$)}\label{alg:AMS}
  \begin{algorithmic}
    \REQUIRE $\graph=(\nodes,\edges), \A\subset\nodes$
     \IF{$v\not\in\phi\left(\A^c\setminus\{v\}\right), \forall v\in \A^c$}
    \STATE Output $\A^c$
    \ELSE
                \FORALL{$v\in\A^c$}
                \IF{$v\in\phi\left(\A^c\setminus\{v\}\right)$}
                \STATE Output $\textbf{AllMaxSetsC}(\graph,\A\cup\{v\})$
                    \ENDIF
             \ENDFOR
     \ENDIF
  \end{algorithmic}
\end{algorithm}


\begin{example}[Butterfly network]
  Figure \ref{butterfly} shows the well-known butterfly network and
  Figure~\ref{mul} shows the corresponding functional dependence
  graph. Nodes are labeled with node numbers and pseudo-variables
  (The sources variables are $Y_1$ and $Y_2$. The $U_i$ are the edge
  variables, carried on links with capacity $C_i$).  Edges in the FDG
  represent the functional dependency due to encoding and decoding
  requirements.
\begin{figure}[htbp]
\centering
  \includegraphics[scale=0.4]{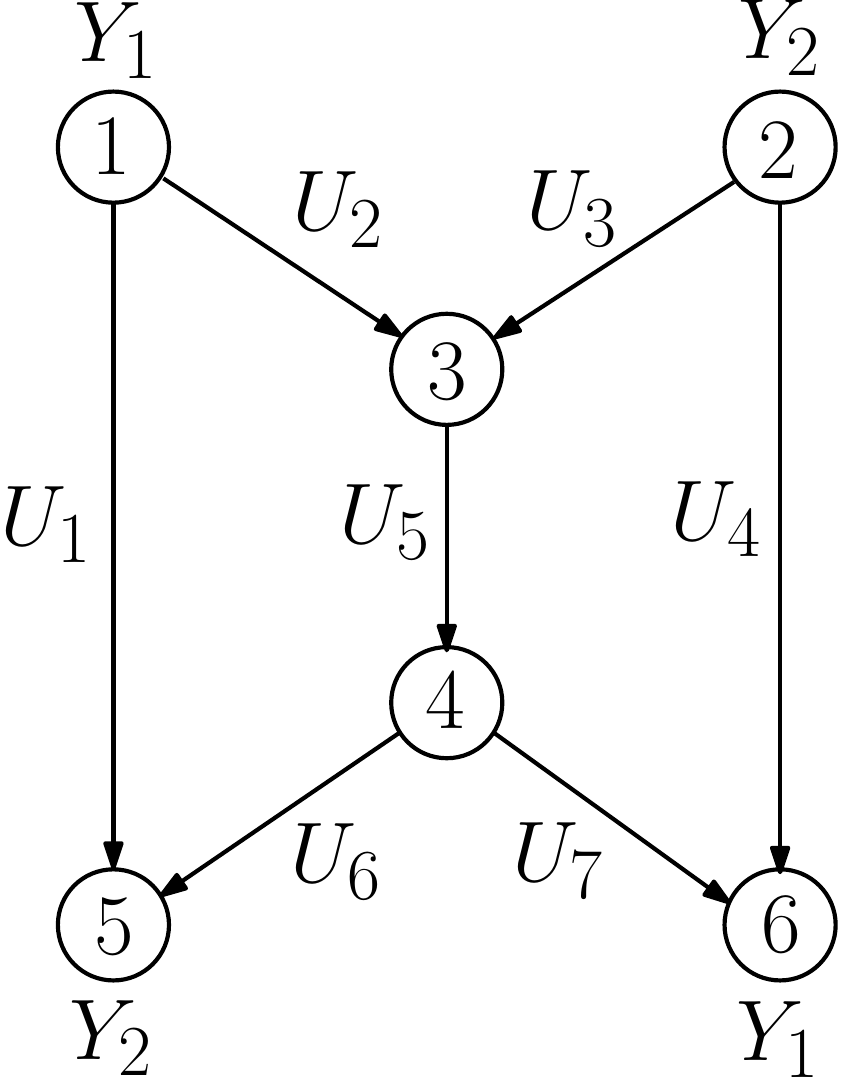}
  \caption{The butterfly network.}\label{butterfly}
\end{figure}
\begin{figure}[htbp]
\centering
  \includegraphics[scale=0.4]{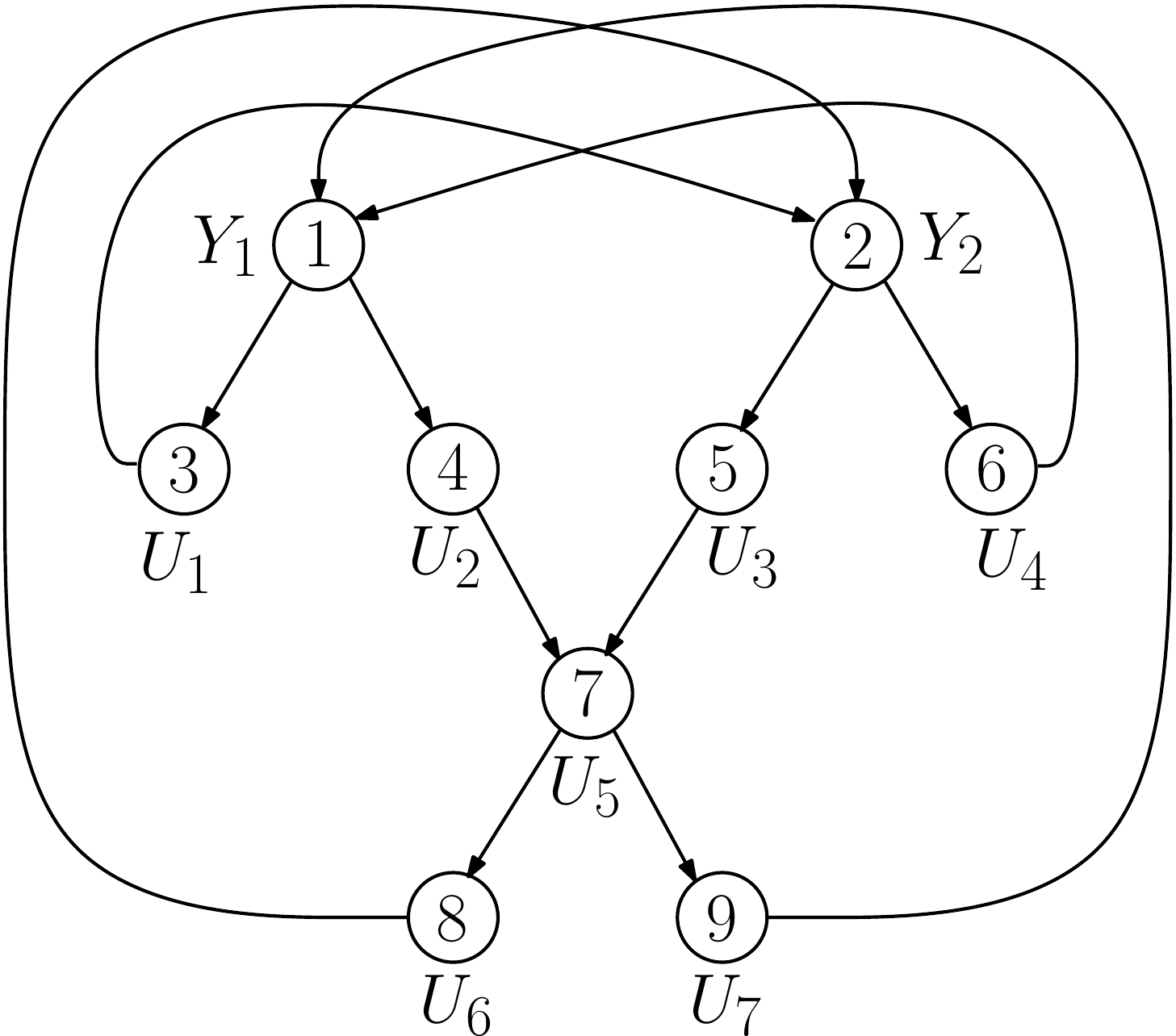}
  \caption{FDG of the butterfly network.}\label{mul}
\end{figure}

The maximal irreducible sets of the cyclic FDG shown in Figure
\ref{mul} are
\begin{multline*} \{1,2\}, \{1,5\}, \{1,7\}, \{1,8\},
\{2,4\}, \{2,7\}, \{2,9\},
\{3,4,5\}, \\
\{3,4,8\}, \{3,7\}, \{3,8,9\}, \{4,5,6\}, \{5,6,9\}, \{6,7\},
\{6,8,9\}.
\end{multline*}
\end{example}

\section{Functional Dependence Bound}\label{sec:bound}
We now give an easily computable outer bound for the total capacity of
a network coding system.
\begin{theorem}[Functional Dependence Bound]\label{thm:mainresult}
  Let $\C_1,\C_2,\C_3,\C_4$ be given network coding constraint sets.
  Let $\graph=(\nodes,\edges)$ be a functional dependence
  graph\footnote{This FDG will be cyclic due to the sink demands
    $\C_3$} on the (source and edge) pseudo-variables
  $Y_\sessions,U_\edges$ with pseudo-entropy function
  $g\in\C_1\cap\C_2\cap\C_3\cap\C_4\cap\Gamma$. Let $\B_M$ be the
  collection of all maximal irreducible sets not containing source
  variables.  Then
\begin{equation*}
  \sum_{s \in \sessions} g(Y_{s}) \leq \min_{\B\in\B_M} \sum_{e :
    U_{e} \in \B} C_{e}.
\end{equation*}
\end{theorem}
\vspace{2mm}
\begin{proof}
Let $\B\in\B_M$, then
\begin{align*}
\sum_{s \in \sessions} g(Y_{s}) &{=} g(Y_\sessions) & \C_1\\
&{=} g(U_{e}: U_{e} \in \B) & \text{Lemma
\ref{lem:equal}}\\
&{\leq} \sum_{U_{e} \in \B} g(U_{e}) &
\text{Subadditivity of $g\in\Gamma$}\\
&{\leq} \sum_{e : U_{e} \in {\B}}
C_{e}. & \C_4
\end{align*}
\end{proof}
Maximal irreducible sets which do not contain source variables
are ``information blockers" from sources to corresponding sinks.
They can be interpreted as information theoretic cuts in in the network.
Note that an improved  bound can in principle be obtained by using
additional properties of $\Gamma$ (rather than just
subadditivity). Similarly, bounds for linear
network codes could be obtained by using $\GammaIn$.

\begin{corollary}
  For single source multicast networks, Theorem \ref{thm:mainresult}
  becomes the max-flow bound \cite[Theorem 11.3]{Yeu02} and hence
  is tight.
\end{corollary}

\begin{example}[Butterfly network]
  The functional dependence bound for the butterfly network of Figure
  \ref{mul} is
  \begin{multline*}
    R_1+R_2 \leq \min
    \{C_3 + C_7,
    C_6+ C_7,
    C_3+ C_4 + C_5,\\
    C_3+ C_4 + C_8,
    C_3+ C_8 + C_9,
    C_4+ C_5 + C_6,\\
    C_5+ C_6 + C_9,
    C_6+ C_8 + C_9\}.
  \end{multline*}
\end{example}

To the best of our knowledge, Theorem \ref{thm:mainresult} is the
tightest bound expression for general multi-source multi-sink network
coding (apart from the computationally infeasible LP bound).  Other
bounds like the network sharing bound \cite{YanYan06} and bounds based
on information dominance \cite{HarKle06} use certain functional
dependencies as their main ingredient. In contrast, Theorem
\ref{thm:mainresult} uses all the functional dependencies due to
network encoding and decoding constraints.

\section{Conclusion}
Explicit characterization and computation of the multi-source network
coding capacity region requires determination of the set of all
entropic vectors $\Gamma^{*}$, which is known to be an extremely hard
problem. The best known outer bound can in principle be computed using
a linear programming approach. In practice this is infeasible due to
an exponential growth in the number of constraints and variables with
the network size.

We gave an abstract definition of a functional dependence graph, which
extends previous notions to accommodate not only cyclic graphs, but
more abstract notions of dependence. In particular we considered
polymatroidal pseudo-entropy functions, and demonstrated an efficient
and systematic method to find all functional dependencies implied by
the given local dependencies.

This led to our main result, which was a new, easily computable
outer bound, based on characterization of all functional
dependencies in networks.  We also show that the proposed bound is
tighter than some known bounds.

\bibliographystyle{ieeetr}

\end{document}